\documentclass[11pt]{article}

\usepackage[utf8]{inputenc}
\usepackage[english]{babel}
\usepackage{kotex}

\usepackage{amsmath, amsthm, amssymb}
\usepackage{mathtools}
\usepackage{bm}       
\usepackage[margin=3.3cm]{geometry}
\usepackage{float}     
\usepackage{algorithm}
\usepackage{algpseudocode}
\usepackage{graphicx}
\usepackage{caption}
\usepackage{subcaption}
\usepackage{booktabs} 
\usepackage{makecell} 
\usepackage{multirow} 
\usepackage[sort&compress, round, semicolon, authoryear]{natbib}
\usepackage[hyphens]{url}
\usepackage[hidelinks]{hyperref}
\usepackage{cleveref}
\usepackage{tikz}
\usetikzlibrary{arrows, positioning, shapes.symbols}

\newtheorem{thm}{Theorem}[]
\newtheorem{lemma}{Lemma}[]

\theoremstyle{remark}
\newtheorem{remark}{Remark}[]

\begin{document}

\title{Dual-Homotopy Framework for Constrained EM Algorithm} 

\author{
{\sc Jisoo Choi}~~and {\sc Hee-Seok Oh}\\
Seoul National University, Seoul 08826, Korea\\
}
\date{ }

\maketitle
\begin{abstract}
We propose a new constrained EM algorithm that is applicable to general constrained estimation problems. The proposed method is based on a novel framework, the `dual-homotopy framework,' which combines deterministic annealing EM with a barrier-based optimization, enabling stable estimation under parameter constraints. Building on this framework, we further introduce an adaptive constrained EM algorithm that preserves likelihood monotonicity, regardless of the underlying distributional form or the specific structure of the constraints. Through simulation studies and a real-data analysis, both under parameter constraints, we demonstrate that the proposed algorithm yields more stable and accurate estimates than existing methods, including the standard EM algorithm. 
\end{abstract}


\newpage
\section{Introduction}

The EM algorithm is a widely used method for likelihood-based inference with latent variables and provides monotonicity and convergence guarantees in unconstrained parameter spaces. However, in many applications where parameters are constrained by structural assumptions or domain knowledge, such as cognitive diagnosis, psychiatric evaluation, or reliability engineering and insurance \citep{wang2002modeling, guPartialIdentifiabilityRestricted2020}, these classical guarantees typically do not apply. These constraints are often introduced to ensure interpretability and practical relevance, yet incorporating arbitrary constraints into the EM framework while preserving likelihood monotonicity remains challenging. In the literature, constrained EM algorithms are well known to suffer from boundary stagnation, which can lead to a loss of monotonic increase in the observed-data likelihood \citep{nettletonConvergencePropertiesEM1999}. Although several constrained EM-type procedures have been developed for specific models, particularly Gaussian mixture models \citep{hathawayConstrainedEmAlgorithm1986, ingrassiaLikelihoodbasedConstrainedAlgorithm2004, ingrassiaConstrainedMonotoneEM2007}, these approaches do not extend to general distributions or arbitrary constraints, and the classical convergence theory of the EM algorithm \citep{dempsterMaximumLikelihoodIncomplete1977a, wuConvergencePropertiesEM1983} does not readily apply to constrained problems.

In this paper, to ensure stable estimation under parameter constraints, we propose a constrained EM framework based on a dual homotopy construction, termed the dual-homotopy EM (DHEM) algorithm. This proposed method combines deterministic annealing EM (DAEM) in the E-step with a barrier method in the M-step within a unified scheme. Furthermore, we develop an adaptive constrained EM algorithm in which parameters are updated according to an explicit acceptance rule that ensures a monotonic increase in the observed-data likelihood. Thus, this adaptive method provides a general constrained learning procedure that preserves likelihood monotonicity in constrained EM algorithms, regardless of the underlying distributional form or the specific structure of the constraints. Specifically, our contributions are summarized as follows:
\begin{itemize}
\item We analyze the latent effects in EM learning and provide a theoretical justification for why DAEM can improve convergence in practice. By studying these latent effects, we demonstrate that DAEM mitigates them, which reduces EM's tendency to converge to local optima.

\item To handle constrained parameters, we incorporate a barrier method into the EM framework. However, the barrier method leads to a pseudo-objective function that can amplify latent effects. To address this, we develop a dual-homotopy EM algorithm (DHEM) by combining it with DAEM.

\item Furthermore, we propose an adaptive DHEM algorithm that provides a clear update rule while ensuring the monotonicity of the observed-data likelihood.
\end{itemize}

The remainder of the paper is organized as follows. 
Section 2 discusses the latent variable in the EM algorithm and reviews deterministic annealing EM. Additionally, we explain how DAEM mitigates latent effects, showing that their impact can be controlled by the annealing parameter. In Section 3, we discuss constrained EM algorithms based on barrier methods. Section 4 proposes a new dual-homotopy EM (DHEM) and an adaptive algorithm that enforces likelihood monotonicity. In Section 5, to evaluate the performance of the proposed method and compare it with existing algorithms, we present examples for various distributions and parameter constraints, including simulation studies based on Gaussian mixture and zero-inflated Poisson models, as well as an analysis of real-world data using a Weibull mixture model. Finally, concluding remarks are offered in Section 6. The {\tt R} codes used for the numerical experiments are available at \url{https://github.com/ddsy999/paper-DHEM}, and the detailed proofs of all theoretical results are provided in the Appendix. 

\section{EM and Deterministic Annealing EM}\label{sec2}
\subsection{EM Algorithm}

Many statistical learning problems are formulated as maximum likelihood estimation (MLE) based on the observed data $X_o$, with the goal of maximizing the observed-data log-likelihood $l_o$. However, directly maximizing $l_o$ is often computationally challenging or intractable. To address this challenge, latent variables $X_m$ are introduced, and the complete data are defined as $X_c = (X_o, X_m)$. The resulting complete-data likelihood is usually easier to handle than $l_o$, even though maximizing it does not always maximize $l_o$. The procedure that uses this complete-data likelihood to iteratively increase $l_o$ is the EM algorithm.


This standard EM algorithm is appealing because it exploits the complete-data likelihood while ensuring a monotonic increase in the observed-data likelihood $l_o$, and, under suitable regularity conditions, it converges to stationary points of the likelihood function \citep{dempsterMaximumLikelihoodIncomplete1977a,wuConvergencePropertiesEM1983}. The EM algorithm alternates between an E-step, which computes the posterior distribution of the latent variables given the current parameter $\theta_0$, i.e., $P(X_m | X_o, \theta_0)$ (the standard posterior), and an M-step, which updates $\theta$ to maximize the conditional expectation of the complete-data log-likelihood, $Q(\theta | \theta_0)$. Updating $\theta$ only to increase $Q(\theta | \theta_0)$ yields a generalized EM (GEM) algorithm. 
%
%
The standard EM algorithm can be understood as an iterative procedure that alternates between parameter estimation and updating the latent distribution. In this view, the M-step corresponds to parameter estimation and often receives the most attention, but the E-step, which determines the latent distribution, is equally important to the proper functioning of the EM algorithm. Although the E-step is often considered straightforward because of its closed-form expression, it plays a critical role in shaping the optimization landscape.

The free-energy perspective provides a direct link between the introduction of latent variables and the EM procedure, allowing the algorithm to be interpreted as a joint optimization over the model parameter $\theta$ and a latent distribution $q$. This perspective highlights that the E-step is not merely a computational step but the solution to an optimization problem, placing it on equal footing with the M-step. In particular, viewing EM as an alternating maximization over $(\theta, q)$ makes it clear that both steps are equally essential to the overall optimization. Specifically, the EM algorithm performs alternating maximization of the free-energy functional defined as
\begin{equation}
\label{eq:standardFreeEnergy}
F(q,\theta) = E_q[\log P(X_c | \theta)] + H(q)
= Q(\theta | q) + H(q),
\end{equation}
where $H(q)$ denotes the entropy of the latent distribution $q$. Maximizing with respect to $q$ for fixed $\theta$ yields the E-step, while maximizing with respect to $\theta$ for fixed $q$ gives the M-step \citep{csiszar1984information,nealViewEmAlgorithm1998a}.

\subsection{Latent Effects}
While introducing latent variables makes likelihood maximization computationally feasible, it also adds complexity. By enlarging the model with unobserved components, the dimension of the problem increases; thus, algorithms operating in this expanded space may behave differently from direct maximization of $l_o$. We refer to these phenomena as latent effects, which can be summarized in two aspects.
First, issues may arise when latent variables are poorly designed. There is no general guideline for constructing latent variables, and the choice is not unique. Although many practitioners construct latent variables for convenience or tractability, there is no guarantee that this aligns with the goal of maximizing the observed-data likelihood. In practice, computational convenience when evaluating the complete-data likelihood is often prioritized, but this does not necessarily imply that it is advantageous for maximizing the observed-data likelihood.

The second issue is that latent variables depend on the previous parameter values $\theta_0$. At each iteration, the latent variables are estimated via $P(X_m|X_o,\theta_0)$ in the E-step, making them directly dependent on the current parameter estimates. These latent variables determine the expected complete-data log-likelihood and thus play a central role in shaping the subsequent M-step updates. Consequently, the previous parameter values influence future updates through the latent variables, which explains why the standard EM algorithm can converge to local maxima when poorly initialized. Moreover, in latent variable models, introducing unobserved components can lead to identifiability issues, as different parameter configurations may yield indistinguishable observed-data distributions unless additional structural restrictions are imposed \citep{guPartialIdentifiabilityRestricted2020}.
%
Several approaches have been proposed to reduce this sensitivity, including improved initialization strategies \citep{biernackiChoosingStartingValues2003} and techniques for escaping local optima. 
In summary, although latent variables enhance computational tractability, they also introduce additional structural complexity that can lead to unstable or suboptimal learning behavior.

\subsection{Deterministic Annealing EM and Reducing Latent Effects}\label{sec:latentEffect}
%

The deterministic annealing EM (DAEM) algorithm, proposed by \citet{uedaDeterministicAnnealingVariant1994,uedaDeterministicAnnealingEM1998}, modifies the E-step of the standard EM algorithm by introducing an annealed (tempered) latent posterior. DAEM has been shown to reduce sensitivity to initial values and mitigate convergence to poor local optima \citep{zamzamiNovelScaledDirichletbased2019,guoParameterizedDeterministicAnnealing2008,bouguilaIntegratingSpatialColor2010}. A comprehensive tutorial review of DAEM methods and their applications is provided by \citet{roseDeterministicAnnealingClustering1998}. In this paper, we provide a more formal, mathematical characterization of the effects of DAEM.

To be specific, in DAEM, the standard posterior is replaced by the annealed posterior, 
\begin{align*}
P_r(X_m | X_o,\theta^{(t)}) 
= \frac{P(X_o,X_m | \theta^{(t)})^r}
{\int P(X_o,X_m | \theta^{(t)})^r \, dX_m},
\end{align*}
where $r \in (0,1]$ is the annealing parameter. This is equal to the standard posterior when $r=1$, while as $r \to 0$, the distribution becomes more diffuse, and latent entropy increases.
%
DAEM performs EM updates with a fixed value of $r$ and then gradually increases $r$ toward 1. The parameter estimate at each annealing level serves as the starting point for the next level. Thus, $r$ serves as a hyperparameter controlling the learning process.

We now explain how DAEM reduces latent effects from three perspectives. First, 
%
%
DAEM admits a natural interpretation in terms of free energy, similar to the standard EM algorithm of (\ref{eq:standardFreeEnergy}).
In general, DAEM maximizes
\[
F(q,\theta|r)
= Q(\theta|q) + \frac{1}{r} H(q),
\quad r \in (0,1].
\]
When $r<1$, the entropy term is weighted more heavily than in the standard EM algorithm. As discussed in \citet{uedaDeterministicAnnealingEM1998}, this induces diffuse latent distributions, thereby reducing the impact of latent effects on estimation and decreasing sensitivity to previous parameter estimates. This entropy-based regularization explains why DAEM reduces sensitivity to initial values and mitigates the risk of converging to local optima.

Second, the ability of DAEM to reduce latent effects can be understood not only through the free-energy interpretation above but also by directly analyzing its objective function. Specifically, we express the DAEM objective, $Q_r(\theta|\theta_0)$, as 
\begin{align}
\label{eq:daemobj}
    Q_r(\theta|\theta_0)
    &=\int P_r(X_m|X_o,\theta_0)\log P(X_c|\theta)\,dX_m\nonumber\\
    &=\log P(X_o|\theta)+\int \log P(X_m|X_o,\theta)P_r(X_m|X_o,\theta_0)dX_m\nonumber\\
    &\triangleq l_o(\theta)+G(\theta,\theta_0|r),
\end{align}
which can be decomposed into the observed-data log-likelihood and a latent-effect term. As shown in (\ref{eq:daemobj}), the objective function admits a decomposition in which the observed-data log-likelihood $l_o(\theta)$ is separated from a remaining term that depends on the latent variable and the previous parameter $(X_m,\theta_0)$. This term is denoted by $G(\theta,\theta_0|r)$ and captures the latent effect. 

\begin{thm}\label{thm:reducingLatentEffect}
For $0 < r \le 1/2$, an upper bound on $|G(\theta,\theta_0|r)|$ is monotone in $r$.
\end{thm}
A proof of Theorem \ref{thm:reducingLatentEffect} is provided in the Appendix. This result shows that the influence of latent effects can be controlled by the annealing parameter. As $r$ decreases, the magnitude of the remainder term shrinks, reducing the distortion caused by previous parameter values.
\begin{remark}
\Cref{thm:reducingLatentEffect} establishes monotonicity only for an upper bound on $|G(\theta,\theta_0 | r)|$ over $0 < r \le 1/2$, rather than for $G(\theta,\theta_0 | r)$ itself. From the free energy perspective, however, any $r<1$ increases the relative weight of the entropy term compared with standard EM, suggesting a reduction in latent effects over the entire range $0<r<1$. Moreover, the derived bound is conservative because it is constructed to hold uniformly over $\theta_0$.
\end{remark}

Finally, another interpretation of DAEM is that it expands the effective search region around $\theta_0$. Ignoring entropy constants, the DAEM objective function is expressed as
\begin{align*}
    Q_r(\theta|\theta_0)
    &= \log P(X_o|\theta)
       - \frac{1}{r} D_{KL}(P_r, \theta_0 \| P_r, \theta)
       - \frac{1}{r} H(P_r, \theta_0)
       + \frac{1}{r}\log Z_r(\theta)
       - \log Z(\theta) \\
    &\simeq l_o(\theta) - \frac{1}{r} D_{KL}(P_r,\theta_0 \| P_r,\theta), \notag
\end{align*}
where $ Z_r(\theta)=\int P(X_c|\theta)^r\,dX_m$, $~Z(\theta)\triangleq L_o(\theta) =\int P(X_c|\theta)\,dX_m$, ~and 
\begin{align*}
D_{\mathrm{KL}}(P_r,\theta_0\|P_r,\theta)&=\int \log\!\left(\frac{P_r(X_m|X_o,\theta_0)}{P_r(X_m|X_o,\theta)}\right)P_r(X_m|X_o,\theta_0)dX_m.
\end{align*}
This divergence term, $D_{KL}$, penalizes large deviations from the previous parameter $\theta_0$, limiting $Q_r(\theta|\theta_0)$ from exploring points far from $\theta_0$ in a single iteration. Consequently, the effective search region is locally restricted. If a suboptimal point lies within this region, the update may be attracted to it rather than moving toward the global optimum. However, Lemma 3 in the Appendix establishes that 
\[
\frac{1}{r} D_{KL}(P_r,\theta_0\|P_r,\theta) = O(r).
\]
Therefore, as $r$ decreases, the impact of this penalty weakens, allowing the algorithm to explore a wider range than with the previous parameter. In contrast, the standard EM algorithm lacks a mechanism to control this dependency.
  
\section{Constrained EM algorithms and Barrier Methods}
The EM algorithm demonstrates stable convergence and effective learning in unconstrained parameter spaces and has been successfully used in many statistical models. However, in many real-world cases, model parameters are constrained by theoretical or empirical considerations. These constraints often represent fundamental assumptions that are expected to hold by design and are therefore taken for granted during modeling, to the extent that they are often not explicitly accounted for during the learning process. Nonetheless, unconstrained learning procedures may violate these assumptions in practice, leading to unstable estimation or even failure of the learning process.

A representative example arises in Gaussian mixture models (GMMs) with a fixed number of components. In this setting, it is often implicitly assumed that the covariance matrices remain nonsingular and that the component means are sufficiently well separated, even though these assumptions are not always explicitly enforced during learning. To prevent singularities caused by component collapse or merging, several constrained EM procedures have been proposed \citep{hathawayConstrainedEmAlgorithm1986,ingrassiaLikelihoodbasedConstrainedAlgorithm2004,ingrassiaConstrainedMonotoneEM2007}. Moreover, constrained optimization techniques that incorporate barrier functions and regularization have been widely adopted to enforce structural properties, such as positive definiteness, and to stabilize estimation results \citep{rothmanPositiveDefiniteEstimators2012,xuePositiveDefiniteL1PenalizedEstimation2012}. Despite these developments, constrained EM algorithms do not always ensure stable or monotone convergence. Research on integrated frameworks that support constrained EM algorithms in more general settings remains limited.

To address these difficulties, we adopt a constrained EM approach based on barrier methods that incorporate constraints into an interior penalty term, thereby transforming a constrained problem into a sequence of unconstrained problems. This approach is widely used because it manages complex constraints while using standard unconstrained optimization techniques \citep{fiaccoNonlinearProgramming1990}. 
In the EM framework, this barrier approach is implemented by modifying the M-step objective function. Specifically, the expected complete-data log-likelihood, $Q(\theta | \theta_0)$, is augmented with a barrier function $B(\theta)$, yielding the pseudo log-likelihood
\[
BQ(\theta,\xi | \theta_0) = Q(\theta | \theta_0) + \xi B(\theta),
\]
where $\xi > 0$ is a barrier parameter that controls the strength of the constraint. The barrier function diverges to $-\infty$ at the boundary of the feasible region, preventing boundary solutions. For a fixed $\xi$, the EM algorithm is run until convergence. As $\xi \to 0$, the barrier's effect diminishes, and the corresponding solution approaches that of the original log-likelihood problem. In this way, $\xi$ serves as a hyperparameter that controls the progression of the constrained EM process.
%
In practice, logarithmic barriers are commonly used as a barrier function $B(\theta)$. For example, under the constraint $C = \{\beta : \beta > -2\}$, one may choose $B(\beta) = \log(\beta + 2)$. 
%
\begin{figure}
    \centering 
    \includegraphics[width=0.8\textwidth]{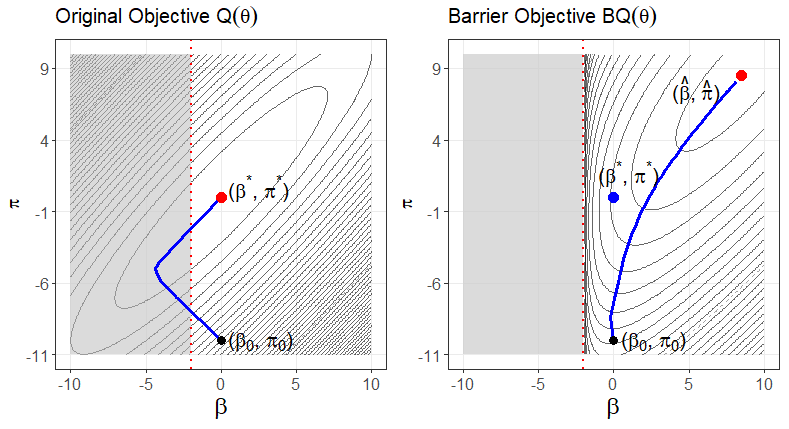}
    \vspace{-7mm}
    \caption{The original objective function $Q(\theta)$ and its solution path (left), and the barrier objective function $BQ(\theta)$ and its solution path.}
    \label{fig:barrierMethod}
\end{figure}
The barrier function enforces learning within the feasible region by imposing constraints, but it distorts the optimization path via a pseudo-log-likelihood. As shown in \Cref{fig:barrierMethod}, the left panel displays the original objective function, $Q(\theta)$, while the right panel shows the barrier objective function, $BQ(\theta)$. The initial value is set to $\beta_0=\beta^\star$. For the original objective $Q(\theta)$, the optimization path may cross the boundary, leading to boundary stagnation when constraints are imposed. By contrast, the barrier objective $BQ(\theta)$ keeps the path within the feasible region by bending the surface away from the boundary. However, because $BQ(\theta)$ is a modified objective, its optimum $(\hat{\beta},\hat{\pi})$ generally differs from the true optimum of $Q(\theta)$. As a result, even though the initial value satisfies $\beta_0=\beta^\star$, the updated parameter can move farther from the true optimum, introducing additional distortion in subsequent iterations.
Because the pseudo log-likelihood $BQ(\theta)$ is used, the resulting parameter estimates satisfy the constraints but may deviate from the true optimum. Moreover, poorly estimated parameters can lead to inaccurate latent-variable estimates in the E-step, which, in turn, influence subsequent updates via the latent-variable effect, as discussed in \Cref{sec:latentEffect}.

\section{Proposed Constrained EM Algorithm}\label{sec:introDHEM}


As discussed in the previous section, although the barrier method ensures that parameter estimates remain within the feasible region, it can amplify the effect of latent variables. This observation motivates the use of DAEM to mitigate the latent effect. Both the barrier method and DAEM rely on hyperparameters and are transformed into the standard EM algorithm in the limiting case. The difference between the two methods lies in the stage at which adjustments are made: DAEM adjusts the E-step, while the barrier method modifies the M-step. We now integrate these two approaches to present the dual-homotopy EM (DHEM) algorithm. Furthermore, we propose an adaptive constrained EM algorithm that updates the hyperparameters according to explicit acceptance rules, guaranteeing that the observed-data likelihood is monotonically increasing.

\subsection{Dual-Homotopy EM Algorithm}
Homotopy describes a continuous deformation between two functions $f, ~g : X \to Y$. A homotopy is a continuous mapping $H(x,t) : X \times [0,1] \to Y$ that satisfies 
\[
H(x,t) = (1-t)f(x) + t\,g(x),~~ 
H(x,0) = f(x),~~ H(x,1) = g(x),
\]
where $t$ parameterizes the deformation. Such constructions are commonly used in continuation methods for tracing solution paths \citep{allgowerIntroductionNumericalContinuation2003}.

Both DAEM and barrier methods allow natural homotopy interpretations. In DAEM, the annealing parameter $r$ gradually deforms the annealed posterior into the standard posterior as $r \to 1$, making $r$ a homotopy parameter in the E-step. In the barrier method, the parameter $\xi$ gradually shifts the penalized objective toward the standard EM objective as $\xi \to 0$, creating a homotopy in the M-step.
%
The dual-homotopy EM (DHEM) algorithm combines these two deformations. Specifically, for constrained estimation, the E-step computes $Q_r(\theta | \theta_0)$ using the annealed posterior, and the M-step maximizes the pseudo log-likelihood function, 
\begin{equation}
\label{dhem}
BQ_{r,\xi}(\theta | \theta_0)= Q_r(\theta | \theta_0) + \xi \cdot B(\theta).
\end{equation}
%
In the early stages, when $\xi$ is large and $r$ is small, a large $\xi$ enforces interior feasibility but can distort the latent variables, while the small value of $r$ reduces sensitivity to such distortions. As learning progresses, $r$ increases while $\xi$ decreases, so DHEM gradually converges to the standard EM algorithm. This coordinated deformation is intended to stabilize learning under parameter constraints. However, because the hyperparameters $(\xi, r)$ are updated on a fixed schedule, mismatches between them may lead to suboptimal coordination, potentially causing instability or misalignment in the learning process.

\subsection{Adaptive DHEM}
We propose an adaptive scheme for selecting the barrier parameter, with a particular focus on preserving the monotonicity of the observed-data likelihood.
For the standard EM algorithm, a monotonic increase in the observed-data likelihood $l_o(\theta)$ and convergence to stationary points are well established \citep{wuConvergencePropertiesEM1983}. However, such guarantees typically do not extend to EM variants or constrained settings. Except in special cases involving specific models and constraint structures \citep{ingrassiaConstrainedMonotoneEM2007}, the monotonicity of $l_o(\theta)$ is not guaranteed. This limitation also applies to the DHEM algorithm.
To ensure monotonicity, we propose an adaptive DHEM algorithm that incorporates an explicit acceptance rule. In this framework, a candidate update is accepted only if the observed-data likelihood $l_o(\theta)$ increases and the parameter constraints are satisfied. By enforcing improvement in $l_o(\theta)$ at each iteration, the adaptive DHEM guarantees that the observed-data likelihood function, $l_o(\theta)$, is monotonically increasing. Additionally, rather than using a fixed schedule for the barrier parameter, we employ an adaptive update rule that reduces the burden of hyperparameter tuning in practice. Moreover, such adaptive strategies are well established in the optimization literature, particularly in interior-point methods, where the barrier parameter is dynamically adjusted using optimality and feasibility measures \citep{nocedal2009adaptive,lange1994adaptive}. With respect to the annealing parameter $r$, if the monotonicity condition is not satisfied, the procedure may terminate before $r$ reaches its final value of 1.

\subsubsection{Sufficient Condition for Monotonicity}
In this section, we derive a sufficient condition for monotonicity of DHEM. To this end, we note that the DAEM objective function, $Q_r(\theta|\theta_0)$, admits the following decomposition, 
\begin{align*}
    Q_r(\theta|\theta_0)
    &= \int \log P(X_o,X_m|\theta)\,P_r(X_m|X_o,\theta_0)\,dX_m \\
    &= \log P(X_o|\theta)
       + \int \log \bigg(
           \frac{P(X_m|X_o,\theta)}{P_r(X_m|X_o,\theta_0)}
         \bigg)
         P_r(X_m|X_o,\theta_0)\,dX_m
       - H(P_r,\theta_0) \\
    &= l_o(\theta)
       - D_{KL}(P_r,\theta_0 \| P,\theta)
       - H(P_r,\theta_0).
\end{align*}
Letting  $\Delta l_o(\theta,\theta_0)= l_o(\theta)-l_o(\theta_0)$, $~\Delta Q_r(\theta,\theta_0)= Q_r(\theta|\theta_0)-Q_r(\theta_0|\theta_0)$, and 
\[
    \Delta D_{KL}(\theta_0 \| \theta, r)= \int \log\Bigg(\frac{P(X_m|X_o,\theta_0)}{P(X_m|X_o,\theta)}\Bigg)P_r(X_m|X_o,\theta_0)\,dX_m,
\]
we have 
\begin{equation}
\label{deltalo}
\Delta l_o(\theta,\theta_0)
=
\Delta Q_r(\theta,\theta_0)
+
\Delta D_{KL}(\theta_0 \| \theta, r).
\end{equation}
%
%
For monotonicity of the observed-data log-likelihood $l_o$, it is required that $\Delta l_o(\theta,\theta_0)\ge0$.

For $r\neq1$, $\Delta Q_r(\theta,\theta_0)\ge0$ still holds by the M-step of DAEM, but the sign of $\Delta D_{KL}(\theta_0\|\theta,r)$ cannot be determined in general. Therefore, the monotonicity of $l_o(\theta)$ is no longer guaranteed. By applying the decomposition in (\ref{deltalo}) to the DHEM objective function of (\ref{dhem}), and defining $\Delta B(\theta,\theta_0)=B(\theta)-B(\theta_0)$, the change in the observed-data likelihood under DHEM can be expressed as
\begin{equation}
\label{deltalodhem}
\Delta l_o(\theta,\theta_0)=\Delta BQ_{r,\xi}(\theta,\theta_0)+\Delta D_{KL}(\theta_0\|\theta,r)-\xi\,\Delta B(\theta,\theta_0). 
\end{equation}
In the M-step of DHEM, $\Delta BQ_{r,\xi}(\theta,\theta_0)\ge0$ is guaranteed. However, monotonicity also requires that the remaining terms be nonnegative, and their signs cannot be determined in general. Searching globally over all triples $(\theta,r,\xi)$ that satisfy both conditions is difficult and computationally expensive. Thus, from the result of (\ref{deltalodhem}), we derive the following sufficient condition that defines a subset of triples ensuring monotonicity, 
\begin{align}\label{eq:suffcondition}
\Delta D_{KL}(\theta_0\|\theta,r)
\;\ge\;
\xi\,|\Delta B(\theta,\theta_0)|. 
\end{align}
Hence, DHEM ensures monotonicity by identifying a triple $(\theta,r,\xi)$ that satisfies (\ref{eq:suffcondition}). To construct an explicit procedure, we first verify that $\Delta D_{KL}(\theta_0\|\theta,r)$ can be made sufficiently positive.

\begin{thm}\label{thm:diffKLlowerbound}
Suppose that $\Delta D_{KL}(\theta_0\|\theta,r)$ is continuous in $r$. For any $\theta_0\neq\theta_1$, there exist $r_1\in(0,1]$ and $\delta(\theta_0,\theta_1)>0$ such that
\[
\Delta D_{KL}(\theta_0\|\theta_1,r)
>
\delta(\theta_1,\theta_0)
\]
for all $r\in[r_1,1]$.
\end{thm}
A proof of \Cref{thm:diffKLlowerbound} is provided in the Appendix. By continuity, the divergence term converges to $D_{KL}(\theta_0\|\theta_1)$ as $r\to1$. Accordingly, the lower bound can be chosen as  
\[
\delta(\theta_1,\theta_0)=\eta\,D_{KL}(\theta_0\|\theta_1),\quad \eta\in[0,1).
\]
Thus, by the inequality in (\ref{eq:suffcondition}), for sufficiently large $r$, the term $\Delta D_{KL}(\theta_0\|\theta_1,r)$ can be made positive. Moreover, the barrier term $\xi|\Delta B(\theta_1,\theta_0)|$ can be controlled by decreasing $\xi$. Therefore, a pair $(r,\xi)$ can be chosen to satisfy the sufficient condition in (\ref{eq:suffcondition}), which guarantees a monotonic increase in the observed-data likelihood.

\subsubsection{Practical Implementation of Sufficient Condition}

The sufficient condition is enforced by a specific procedure for selecting the triple $(\theta, r, \xi)$. First, a candidate pair $(\theta_{\mathrm{cand}}, r)$ is obtained via the M-step of DHEM. The candidate then undergoes two acceptance rules. The first rule concerns the annealing parameter $r$,
\begin{align*}
    \text{(1st acceptance rule)}\quad
    \Delta D_{KL}(\theta_0\|\theta_{\mathrm{cand}},r)
    \ge
    \delta(\theta_{\mathrm{cand}},\theta_0).
\end{align*}
According to \Cref{thm:diffKLlowerbound}, there exists a sufficiently large range of $r$ for which this condition holds. If the rule is not satisfied at the current $r$, the algorithm increases $r$ and repeats the step. Once the first acceptance rule is satisfied, the second rule must be verified,
\begin{align*}
    \text{(2nd acceptance rule)}\quad
    \delta(\theta_{\mathrm{cand}},\theta_0)
    \ge
    \xi\,|\Delta B(\theta_{\mathrm{cand}},\theta_0)|.
\end{align*}
If this condition fails, the barrier parameter is reduced to
\[
\xi'
=
\min\!\left\{
\xi,\,
\frac{\delta(\theta_0,\theta_{\mathrm{cand}})}
     {|\Delta B(\theta_0,\theta_{\mathrm{cand}})|}
\right\}.
\]
If both acceptance rules are satisfied, the triple $(\theta_{\mathrm{cand}},r,\xi')$ ensures the monotonicity of $l_o$. The update is then accepted by setting $\theta_1=\theta_{\mathrm{cand}}$, and the algorithm continues with $(\theta_1,r,\xi')$ as a new state. Therefore, each accepted update meets the sufficient condition for a monotone increase in the observed-data likelihood.
In summary, although the monotonicity of the barrier-based constrained EM can be characterized theoretically, directly searching over $(\theta,r,\xi)$ is impractical. The proposed adaptive DHEM algorithm replaces this global search with a sequential two-stage acceptance process. By enforcing sufficient conditions step by step, it guarantees monotonicity while remaining computationally feasible.

\subsubsection{Global Convergence Theorem with Adaptive DHEM}

In the previous section, we established sufficient conditions for the monotonicity of the observed-data log-likelihood and introduced the corresponding acceptance rules. We now show that the proposed adaptive DHEM algorithm satisfies the conditions of Zangwill’s global convergence theorem (GCT) \citep{zangwill1969nonlinear}, thereby ensuring monotonic increase in the observed-data likelihood and convergence to a stationary point under constraints. Specifically, we verify that the two acceptance rules, which adaptively update the triple $(\theta,r,\xi)$, induce a set-valued mapping whose projection onto the parameter component $\theta$ satisfies the GCT conditions.

We now introduce the set-valued maps as follows. The GEM update, for a fixed $(r,\xi)$, defines
\begin{align*}
M(\theta,r,\xi)
&=
\Big\{
\theta' :
BQ_{r,\xi}(\theta' |\theta)
\ge
BQ_{r,\xi}(\theta |\theta)
\Big\}.
\end{align*}
The first acceptance rule, which establishes a lower bound on the divergence, induces a mapping
\begin{align*}
A_1(\theta,r)
&=
\Big\{
(\theta',r') :
\Delta D_{KL}(\theta \| \theta', r') \ge \delta(\theta',\theta),
\ r' \ge r
\Big\}.
\end{align*}
The second acceptance rule, which controls the barrier contribution, yields
\begin{align*}
A_2(\theta,\xi)
&=
\Big\{
(\theta',\xi') :
\delta(\theta',\theta) \ge \xi' |\Delta B(\theta',\theta)|,
\ \xi' \le \xi
\Big\}.
\end{align*}
Combining these components, the full adaptive DHEM update is represented by a set-valued map,
\[
T(\theta,r,\xi)
=
\Big\{
(\theta',r',\xi') :
\theta' \in M(\theta,r',\xi'),\,
(\theta',r') \in A_1(\theta,r),\,
(\theta',\xi') \in A_2(\theta,\xi)
\Big\}.
\]
\begin{thm}[Global convergence of adaptive DHEM]
\label{thm:gct_adaptive_dhem}
Suppose that the following conditions hold: (i) The parameter space $\Theta = \{\theta | l_o(\theta) \ge l_o(\theta_0)\}$ is compact, and $r \in [r_0,1]$ and $\xi \in [0,\xi_0]$ are confined to compact intervals. (ii) The function $BQ_{r,\xi}(\theta'|\theta)$ is continuous in $(\theta',\theta,r,\xi)$. (iii) The divergence $\Delta D_{KL}(\theta \| \theta', r)$ is continuous in $(\theta,\theta',r)$. (iv) The function $\delta(\theta',\theta)$ is continuous in $(\theta',\theta)$. (v) The barrier difference $\Delta B(\theta',\theta)$ is continuous in $(\theta',\theta)$. 
Define the solution set as
$
\Gamma=\{\theta \in \Theta | 0 \in \partial l_o(\theta)\}.
$
Then, the following statements hold:
\begin{enumerate}
    \item The mapping $T(\theta,r,\xi)$ is closed on the compact domain $\Theta \times [r_0,1] \times [0,\xi_0]$.
    \item The induced projection map, 
    $
    T_\Theta(\theta,r,\xi)=\{\theta' : (\theta',r',\xi') \in T(\theta,r,\xi)\},
    $
    is also closed and shares the same solution set $\Gamma$.
    \item The observed-data log-likelihood $l_o(\theta)$ serves as an ascent function, i.e., for any $\theta' \in T_\Theta(\theta,r,\xi)$,
    $
    l_o(\theta') \ge l_o(\theta).
    $
\end{enumerate}
Therefore, any sequence $\{\theta^{(t)}\}$ generated by the adaptive DHEM algorithm satisfies
\[
l_o(\theta^{(t+1)}) \ge l_o(\theta^{(t)}),
\]
and every limit point of $\{\theta^{(t)}\}$ belongs to $\Gamma$. In particular, the sequence converges to a stationary point of $l_o(\theta)$ under the given constraints.
\end{thm}

\subsubsection{Initial Setting of Hyperparameters}
In the proposed method, the initial settings of the two hyperparameters, namely the annealing parameter $r$ and the barrier parameter $\xi$, are crucial to learning performance. Because $r\in[0,1]$, the annealing parameter space is relatively limited. According to \Cref{thm:reducingLatentEffect}, it is desirable to start with $r<0.5$ to effectively reduce the latent effect. In the subsequent data analysis, we set $r_{\texttt{init}}=0.1$. 

In contrast, the barrier parameter takes values in $[0,\infty)$, and its appropriate magnitude depends on the specific form of the barrier function $B(\theta)$, which necessitates a principled guideline. 
The key principle for selecting $\xi_{\texttt{init}}$ is that, at the initial point $\theta_0$, the learning direction should be primarily governed by the surrogate function $Q_r(\theta|\theta_0)$ rather than the barrier term. That is, in $BQ_{r,\xi}(\theta|\theta_0)=Q_r(\theta|\theta_0)+\xi B(\theta)$, the initial gradient should be dominated by $\nabla Q_r(\theta_0|\theta_0)$ rather than $\xi\nabla B(\theta_0)$. This requirement can be expressed as $||\nabla Q_r(\theta_0|\theta_0)||>\xi_{\texttt{init}}\cdot||\nabla B(\theta_0)||$, and selecting $\xi_{\texttt{init}}$ within this range ensures that the learning proceeds in an appropriate direction at the initial stage. Moreover, since $\theta_0$ is chosen within the feasible region, $\xi_{\texttt{init}}$ can be determined afterward, making this criterion practically applicable.
More specifically, by introducing $\tau\in(0,1)$, one can directly control the dominance level of $Q_r$, which determines the relative influence of the surrogate term in the initial learning phase. Accordingly, we impose the condition $||\nabla B(\theta_0)||\cdot\xi_{\texttt{init}} \le \tau\cdot||\nabla Q_r(\theta_0|\theta_0)||$, which leads to 
\[
\xi_{\texttt{init}}=\tau \frac{||\nabla Q_r(\theta_0|\theta_0)||}{||\nabla B(\theta_0)||}. 
\]
This choice prevents the barrier term from exerting excessive influence while allowing the constraint to be incorporated progressively. Moreover, this formulation defines $\xi_{\texttt{init}}$ as a quantity jointly driven by the data, the barrier function $B(\theta)$, and the initial parameter $\theta_0$.

\section{Numerical Experiments}\label{sec:DataAnalysis}

To evaluate the performance of the proposed adaptive DHEM algorithms, we conduct two simulation studies and a real-data analysis. The simulation studies use Gaussian mixture models (GMMs) and zero-inflated Poisson (ZIP) models, and the real-data analysis examines a Weibull mixture model to capture bathtub-shaped hazard rates. These models are representative of mixture-based frameworks, which are well known for issues such as identifiability and component ambiguity \citep{teicher1963identifiability}. In all cases, we compare five methods, including the proposed ones: standard EM (EM), deterministic annealing EM (DAEM), barrier method EM, dual-homotopy EM (DHEM), and adaptive DHEM.

\subsection{Gaussian Mixture Model}

Gaussian mixture models are a classic application of the EM algorithm, yet they are prone to issues such as component collapse and covariance singularity during estimation. These problems are exacerbated in higher dimensions or with poor initialization, causing unstable estimates and making parameters less identifiable. In this study, we address these issues by ensuring sufficient separation between component means, aligning with the statistical need for mixture components to represent distinct clusters. The identifiability of mixture models has been widely discussed in the literature \citep{chenOptimalRateConvergence1995,teicher1963identifiability}. Building on this, we incorporate component-wise distances via a barrier formulation to enforce separation. 

The underlying assumption that the model consists of $K$ distinct components induces an implicit feasible parameter space, 
$
\mathcal{M}=
\mathbb{R}^{K\cdot d}
\cap
\left\{
(\mu_1,\ldots,\mu_K)
: |\mu_k - \mu_l| > \epsilon  \text{ for all } k \neq l
\right\}, 
$
which requires the component means to be mutually distinct.
Early work on constrained Gaussian mixture models has largely focused on preventing covariance singularity \citep{hathawayConstrainedEmAlgorithm1986, ingrassiaLikelihoodbasedConstrainedAlgorithm2004, ingrassiaConstrainedMonotoneEM2007}. More recently, similar ideas have been explored using penalty-based approaches \citep{manoleEstimatingNumberComponents2021}. In this work, we adopt a barrier formulation that imposes constraints based on Mahalanobis distances between component means and implements them via a log-barrier function.
%
With deterministic annealing and a log-barrier function, the objective function of GMM has the form, 
\[
BQ_{r,\xi}(\theta|\theta^{(t)})=Q_r(\theta|\theta^{(t)})+\xi B(\theta),
\]
for $\theta=\{(\pi_k,\mu_k,\Sigma_k): k=1,\cdots,K\}$, where 
\[
Q_r(\theta|\theta_0)=\sum_{i=1}^N \sum_{k=1}^K
Z_{ik}^{(r)}(\theta_0)
\big(
\log \pi_k + \log \mathcal N(x_i|\mu_k,\Sigma_k)
\big),\\
\]
and $Z_{ik}^{(r)}(\theta_0)$ denotes the annealed responsibility. The constrained objective function is defined as $BQ_{r,\xi}(\theta|\theta^{(t)})$ with a log-barrier function $B(\theta)$ to enforce a minimum separation between component means.

Unlike conventional barrier methods, we use a component-wise barrier formulation. Specifically, for updating $\mu_k$, we define the component-wise barrier as 
\[
B^{[k]}(\theta)=\log\big(d_k^2(\mu_k)-\delta\big),
\]
where 
$
d_k^2(\mu_k)=\sum_{l\neq k}(\mu_k-\mu_l)^\top \Sigma_k^{-1}(\mu_k-\mu_l),
$
and $k$ denotes the index of the component currently being updated. During the update of the $k$th component, the barrier term $B^{[k]}(\theta)$ is used, while the other components are treated as fixed. The parameter update proceeds sequentially, first updating $\mu_k$ and then updating $\Sigma_k$. Specifically, for fixed \(r\) and \(\xi\), we first compute \(\mu_k\) by solving the first-order optimality condition of the barrier-augmented objective, and then update \(\Sigma_k\) in closed form using the updated means.

The update of \(\mu_k\) is obtained by solving the stationarity condition of \(BQ_{r,\xi}(\theta|\theta^{(t)})\). Let
$
N_k = \sum_{i=1}^N Z_{ik}^{(r)}(\theta^{(t)})
$
and 
$ 
s_k = \sum_{i=1}^N Z_{ik}^{(r)}(\theta^{(t)}) x_i.
$
Then, the first-order condition with respect to $\mu_k$ is given by
\begin{align*}
0=\frac{\partial}{\partial \mu_k} BQ_{r,\xi}(\theta|\theta^{(t)})&=\Sigma_k^{-1}\big(s_k - N_k \mu_k\big)+\xi \cdot\frac{2}{d_k^2(\mu_k)-\delta}\sum_{l\neq k}\Sigma_k^{-1}(\mu_k-\mu_l)\\
&=\Sigma_k^{-1}\bigg( s_k - N_k \mu_k + \xi \cdot\frac{2}{d_k^2(\mu_k)-\delta}\sum_{l\neq k}(\mu_k-\mu_l)\bigg).
\end{align*}
Since $\Sigma_k^{-1}$ is positive definite, this condition is equivalent to
\[
0=s_k - N_k \mu_k + \xi \cdot\frac{2}{d_k^2(\mu_k)-\delta}\sum_{l\neq k}(\mu_k-\mu_l).
\]
To derive this equation, we introduce a component-wise barrier function. This nonlinear equation defines a candidate update $\mu_k^{\text{prop}}$, computed numerically via a Gauss-Seidel scheme; that is, each $\mu_k$ is updated sequentially while the remaining components $\{\mu_l\}_{l\ne k}$ are held fixed.

Given the updated means $\{\mu_k\}$, the covariance matrices are updated in closed form as
\begin{align*}
\Sigma_k=\frac{1}{N_k}\sum_{i=1}^NZ_{ik}^{(r)}(\theta^{(t)})\,(x_i-\mu_k)(x_i-\mu_k)^\top.
\end{align*}
The resulting parameters $\{\mu_k,\Sigma_k\}$ satisfy the generalized EM (GEM) condition,
\[
BQ_{r,\xi}(\theta^{(\text{prop})}|\theta^{(t)}) \ge BQ_{r,\xi}(\theta^{(t)}|\theta^{(t)}).
\]
However, the candidate $\{\mu_k^{\text{prop}}\}$ is not accepted directly. Although it satisfies the GEM condition, it may violate the feasibility constraint, whereas the previous parameter $\mu_k^{(t)}$ is guaranteed to be feasible. To address this, we employ a backtracking damping procedure. Specifically, we construct intermediate candidates
\[
\mu_k^{(\alpha)} = (1-\alpha)\mu_k^{(t)} + \alpha \mu_k^{\text{prop}}, \quad \alpha \in (0,1],
\]
and recompute $\Sigma_k^{(\alpha)}$ accordingly. Among these candidates, we select the first one that is feasible under the barrier constraint and satisfies the GEM condition.
Since the previous parameters $\{\mu_k^{(t)},\Sigma_k^{(t)}\}_{k=1,\cdots,K}$ are feasible, the backtracking procedure is likely to find candidates that are both feasible and satisfy the GEM condition. If no such candidate is found, the update is rejected, the previous parameter \(\theta^{(t)}\) is kept, and the procedure proceeds by increasing the annealing parameter \(r\) and repeating the update.

The initial barrier parameter \(\xi_{\mathrm{init}}\) in GMM is determined by the relative magnitudes of the surrogate and barrier gradients. Because we adopt a profiling approach with respect to \(\mu_k\), the derivative with respect to \(\mu_k\) can be computed explicitly, enabling a direct comparison of gradient magnitudes.
The first-order condition for \(\mu_k\) is given by
\[
0=s_k - N_k \mu_k+\xi \cdot\frac{2}{d_k^2(\mu_k)-\delta}\sum_{l\neq k}(\mu_k-\mu_l),
\]
which can be interpreted as the sum of the surrogate gradient term \(s_k - N_k \mu_k\) and the barrier-induced term. To ensure that early iterations are primarily driven by the surrogate objective \(Q_r\), we impose the following condition at the initial point \(\mu_k^{(0)}\), 
\[
\xi\Big\|\frac{2}{D_k(\mu_k^{(0)})}\,u_k(\mu_k^{(0)})\Big\|\le\tau \,\|s_k - N_k \mu_k^{(0)}\|,
\]
where \(u_k(\mu_k)=\sum_{l\neq k}(\mu_k-\mu_l)\), \(D_k(\mu_k)=d_k^2(\mu_k)-\delta\), and \(\tau\in(0,1)\). 
This leads to the following initialization, 
\[
\xi_{\mathrm{init}}=\frac{\tau}{2}\min_{k}\Big\{\frac{\|s_k - N_k \mu_k^{(0)}\|\,D_k(\mu_k^{(0)})}{\|u_k(\mu_k^{(0)})\|}\Big\}.
\]

We conduct 500 simulation replications and compare five algorithms: EM, DAEM, barrier-method EM, DHEM, and adaptive DHEM, using the same initialization and convergence criteria across all methods. %
%
%
The true parameters used in the simulation are $n=100$, 
$\pi=(0.2,0.3,0.5)$,
$\mu_1=(-1,1,2)^\top$,
$\mu_2=(1,1,0.5)^\top$, and
$\mu_3=(2,0,-2)^\top$.
The covariance matrices are specified to have distinct non-diagonal structures, given by
\[
\Sigma_1=\begin{pmatrix}
1&0.3&-0.2\\
0.3&1&0.1\\
-0.2&0.1&1
\end{pmatrix},\ 
\Sigma_2=\begin{pmatrix}
2&0.18&-0.25\\
0.18&0.1&0.06\\
-0.25&0.06&0.5
\end{pmatrix},\ 
\Sigma_3=\begin{pmatrix}
0.5&-0.08&0.22\\
-0.08&0.1&-0.12\\
0.22&-0.12&2
\end{pmatrix}.
\]


\begin{table}[!h]
\centering
\caption{Averaged estimation error values (their standard deviations) for $\pi$, $\mu$, $\Sigma$ by five methods.}
        \vspace{-3mm}
\label{tab:gmm_err_mean_sd}
{\small
\begin{tabular}{lcccc}
\toprule
Method & $\pi$ & $\mu$ & $\Sigma$ & success \\
\midrule
Adaptive DHEM & 0.066 (0.044) & 0.589 (0.310) & 1.376 (0.918) & 0.960 \\
Barrier & 0.121 (0.081) & 0.933 (0.565) & 1.623 (1.012) & 0.880 \\
DAEM & 0.100 (0.031) & 1.625 (0.694) & 3.276 (1.454) & 1.000 \\
DHEM & 0.099 (0.065) & 1.079 (0.626) & 2.494 (1.502) & 1.000 \\
EM & 0.121 (0.081) & 0.933 (0.562) & 1.623 (1.002) & 0.900 \\
\bottomrule
\end{tabular}}
\end{table}

Estimation accuracy is evaluated using the $L_1$ norm for mixture weights, the $L_2$ norm for component means, and the Frobenius norm for covariance matrices. 
%
\Cref{tab:gmm_err_mean_sd} presents the average estimation errors for $\pi$, $\mu$, and $\Sigma$, along with their standard deviations, and the success rates of parameter convergence. As shown, the standard EM algorithm and the Barrier EM algorithm often exhibit unstable behavior, reflected in their relatively low success rates. The standard EM algorithm suffers from component collapse, while the Barrier EM algorithm shows unstable variance due to significant distortion of the latent structure.
Because the initial parameters are randomly initialized, there is a risk of convergence failure across all methods. Among them, DAEM exhibits the most stable convergence behavior, achieving a perfect success rate across replications. However, as discussed in \citet{yuConvergenceParameterSelection2018}, it suffers from a merging effect in which component means collapse toward one another, leading to relatively large estimation errors. 
The standard EM and Barrier EM algorithms are more prone to convergence failure due to distortion caused by latent variables, as reflected in their lower success rates. The DHEM algorithm shows intermediate performance because its effectiveness depends on a proper balance between the two hyperparameters; when this balance is not well maintained under a fixed schedule, performance may degrade. 
Overall, the adaptive DHEM consistently achieves both high accuracy and stable convergence, providing the most reliable performance among the compared methods.

\subsection{Zero-Inflated Poisson Model}
Zero-inflated Poisson (ZIP) models are commonly used for count data with many zeros, such as insurance claim counts or disease incidence data, and are typically estimated using the EM algorithm. In the ZIP model, the probability of structural zeros is determined by a zero-inflation parameter $\pi$, while the remaining observations follow a Poisson distribution with mean $\lambda$. 
The ZIP model is given by
\[
P(Y=0)=\pi+(1-\pi)e^{-\lambda}, \qquad
P(Y=y)=(1-\pi)\frac{\lambda^y e^{-\lambda}}{y!}, \quad y\ge1.
\]
When the Poisson mean $\lambda$ is very small and the zero-inflation probability $\pi$ is high, an identifiability problem may arise. In this case, excess zeros can be explained by both the Poisson and zero-inflation components, making it difficult to distinguish between these mechanisms and leading to potential ambiguity in parameter estimation \citep{klein2015bayesian}.

In real-world applications such as insurance data, the proportion of zero-inflated observations is often high, while the Poisson event rate is very low, reflecting the aforementioned setting and leading to the same identifiability issue. Consequently, prior research or domain knowledge often sets a lower bound on $\pi$ to incorporate prior information.
To reflect this, we impose a constraint on $\pi$ and implement it via a log-barrier function. To avoid imposing overly strong prior information relative to $\pi_{\texttt{true}}$, we adopt a relatively mild constraint. Although it is known that $\pi_{\texttt{true}}$ is high, we impose only a weakly informative lower bound. For example, the constraint $\pi>0.5$ can be enforced with the barrier $B(\pi)=\log(\pi-0.5)$, which prevents $\pi$ from being driven toward unrealistically small values during estimation without overly favoring the true value. This induces a constrained parameter space, $\Theta=\left\{ (\pi,\lambda) \in (\pi_{\min},1)\times(0,\infty) \right\}$. 

For the ZIP model, the initial barrier parameter \(\xi_{\mathrm{init}}\) is selected by comparing the surrogate and barrier gradients with respect to \(\pi\). To enforce the constraint \(\pi>\pi_{\min}\), we use the log-barrier, 
$
B(\pi)=\log(\pi-\pi_{\min}).  
$
Then, the surrogate gradient is given by
\[
\frac{\partial Q_r}{\partial \pi}=\frac{\sum_i (1-Z_i^{(0)})}{\pi}-\frac{\sum_i Z_i^{(0)}}{1-\pi}.
\]
To ensure that the initial updates are primarily driven by the surrogate objective \(Q_r\), we impose the condition
\[
\xi \Big|\frac{1}{\pi^{(0)}-\pi_{\min}}\Big|\le\tau \Big|\frac{\partial Q_r}{\partial \pi}(\pi^{(0)})\Big|,
\]
which yields
\[
\xi_{\mathrm{init}}=\tau\Big|\frac{\partial Q_r}{\partial \pi}(\pi^{(0)})\Big|(\pi^{(0)}-\pi_{\min}).
\]

The simulation study is repeated 200 times with a sample size of 10,000 in each repetition. The true parameters are set to $\pi_{\texttt{true}}=0.99$ and $\lambda_{\texttt{true}}=0.3$, while the initial values are fixed at $\pi_{\texttt{init}}=0.7$ and $\lambda_{\texttt{init}}=1$. For the hyperparameter, the annealing parameter is initialized at $r_{\texttt{init}}=0.1$. The barrier function is defined as $\log(\pi-0.5)$. 
%
\begin{figure}
    \centering
    \begin{subfigure}[t]{0.49\textwidth}
        \centering
        \includegraphics[width=\textwidth]{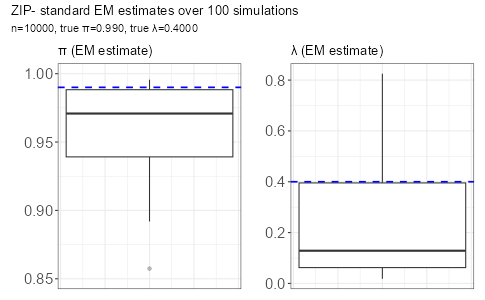}
              \caption{Standard EM}
        \label{fig:ZIP_stamdEMresult}
    \end{subfigure}
    \hfill
    \begin{subfigure}[t]{0.49\textwidth}
        \centering
        \includegraphics[width=\textwidth]{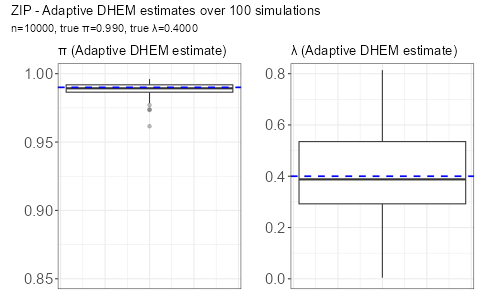}
        \caption{Adaptive DHEM}
        \label{fig:ZIP_adaptiveDHEMresult}
    \end{subfigure}
    \vspace{-5mm}
    \caption{Boxplots of estimated values for $\pi$ and $\lambda$ by standard EM and adaptive DHEM.}
    \label{fig:ZIP_EM_vs_adaptiveDHEM}
\end{figure}

\begin{table}[!h]
\centering
\caption{Averaged estimation error values (their standard deviations) for $\pi$, $\lambda$ by five methods.}
        \vspace{-3mm}
\label{tab:zip_err_mean_sd}
{\small
\begin{tabular}{lcc}
\toprule
Method & $\pi$  & $\lambda$   \\
\midrule
Adaptive DHEM & $-0.003\,(0.007)$ & $-0.010\,(0.123)$ \\
Barrier      & $-0.003\,(0.007)$ & $-0.010\,(0.123)$ \\
DAEM         & $-0.431\,(0.030)$ & $-0.293\,(0.002)$ \\
DHEM         & $-0.431\,(0.030)$ & $-0.293\,(0.002)$ \\
EM           & $\phantom{-}0.270\,(0.037)$ & $\phantom{-}3.171\,(0.274)$ \\
\bottomrule
\end{tabular}}
\end{table}

A comparison of the learning trajectories of the standard EM and the adaptive DHEM is shown in \Cref{fig:ZIP_EM_vs_adaptiveDHEM}, where the blue dashed lines indicate the true parameter values. Panel (a) shows that both $\lambda$ and $\pi$ converge to values significantly lower than their true values under the standard EM. In contrast, panel (b) shows that the adaptive DHEM prevents $\pi$ from dropping below the imposed lower bound. As a result, the estimates of $\pi$ remain bounded away from zero and stay close to the true parameter value, yielding more stable and interpretable estimates.
From \Cref{tab:zip_err_mean_sd}, the adaptive DHEM and the barrier method exhibit nearly identical performance, with matching estimation errors for both $\pi$ and $\lambda$. This indicates that the primary source of improvement is the barrier mechanism rather than the annealing component. In contrast, DAEM and DHEM perform nearly identically, but both perform worse than the barrier methods and adaptive DHEM.
This behavior stems from the ZIP model's simple latent structure. For $Y \ge 1$, the latent variable is deterministically assigned to the Poisson component ($Z = 0$), leaving little room for annealing to alter the latent distribution. Consequently, annealing introduces distortion without meaningful benefit, leading to inferior performance of DAEM and DHEM.

\subsection{Real Data Analysis: Bathtub Curve}

In reliability engineering, the `bathtub curve' is a canonical model for the life cycle of manufactured products; it is characterized by a hazard rate that initially decreases, remains nearly constant for a period, and then increases again. This pattern is often represented by a mixture of three Weibull distributions, corresponding to infant mortality, random failure, and wear-out failure. To capture this structure, the shape parameters of the Weibull components must satisfy the ordered constraints,
$
\beta_1 < 1,~ \beta_2 = 1,~ \beta_3 > 1.$
These constraints are more than numerical limits; they also reflect the physical meaning of each failure regime. 

Given the Weibull density
\begin{align*}
f(t)=\frac{\beta}{\alpha}\bigg(\frac{t}{\alpha}\bigg)^{\beta-1}\exp\!\big(-(t/\alpha)^\beta\big)=
\beta \lambda t^{\beta-1}\exp(-\lambda t^\beta),
\qquad \alpha>0,\ \beta>0, \ \lambda=\alpha^{-\beta},
\end{align*}
we consider a constrained EM estimation problem in which the shape parameters $\beta$ are subject to the ordering constraints. For analysis, we use a well-known real dataset exhibiting a bathtub-shaped hazard function, consisting of failure times of electronic devices in \citet{aarsetHowIdentifyBathtub1987}. The initial shape parameters are set to $\beta_{\texttt{init}}=(0.5,\,1,\,2)$, which satisfy the ordering constraints. The annealing parameter is initialized at $r_{\texttt{init}}=0.1$, and the initial barrier parameter \(\xi_{\mathrm{init}}\) is chosen by comparing the surrogate gradient and the barrier gradient with respect to the shape parameters. We consider constraints of the form \(\beta_1 \in (0,1)\) and \(\beta_3 > 1\), and introduce log-barrier functions,  
$
B(\beta_1)=\log \beta_1 + \log(1-\beta_1),~B(\beta_3)=\log(\beta_3-1).
$
The corresponding gradients are given by
$
\frac{\partial Q_r}{\partial \beta_j},~ j\in\{1,3\}.
$
To ensure that the initial updates are primarily driven by \(Q_r\), we impose
\[
\xi \Big|\frac{\partial B}{\partial \beta_j}(\beta_j^{(0)})\Big|\le\tau \Big|\frac{\partial Q_r}{\partial \beta_j}(\beta_j^{(0)})\Big|,\quad j\in\{1,3\},
\]
which yields the initialization, 
\[
\xi_{\mathrm{init}}=\tau\min\Big\{\Big|\frac{\partial Q_r}{\partial \beta_1}(\beta_1^{(0)})\Big|\min\big(\beta_1^{(0)},\,1-\beta_1^{(0)}\big),\;\Big|\frac{\partial Q_r}{\partial \beta_3}(\beta_3^{(0)})\Big|(\beta_3^{(0)}-1)\Big\}.
\]

\begin{table}
\centering
\caption{Converged parameter estimates obtained by five methods}
\vspace{-3mm}
\label{tab:converged_params_math_header}
\scriptsize
\setlength{\tabcolsep}{3pt}
\begin{tabular}{lrrrrrrrrrr}
\hline
Method 
& $\pi_1$ & $\pi_2$ & $\pi_3$
& $\lambda_1$ & $\lambda_2$ & $\lambda_3$
& $\beta_1$ & $\beta_3$
& $\nabla Q(\beta_1)$ & $\nabla Q(\beta_3)$ \\
\hline
EM       
& 0.13 & 0.62 & 0.26
& 1.04 & 0.026 & $2.9\times 10^{-152}$
& 1.56 & 78.57
& $6.7\times 10^{-16}$ & $-1.4\times 10^{-14}$ \\

DAEM     
& 0.33 & 0.33 & 0.33
& 0.027 & 0.022 & $2.7\times 10^{-2}$
& 0.95 & 0.95
& $2.7\times 10^{-13}$ & $-4.3\times 10^{-14}$ \\

Barrier  
& 0.13 & 0.62 & 0.26
& 1.05 & 0.026 & $2.9\times 10^{-152}$
& 1.00 & 78.57
& $2.1$ & $-1.3\times 10^{-10}$ \\

DHEM     
& 0.13 & 0.62 & 0.26
& 1.05 & 0.026 & $2.9\times 10^{-152}$
& 1.00 & 78.57
& $2.1$ & $-1.3\times 10^{-10}$ \\

Adaptive DHEM 
& 0.24 & 0.51 & 0.25
& 0.26 & 0.025 & $2.5\times 10^{-151}$
& 0.57 & 78.09
& $4.4\times 10^{-6}$ & $-1.1\times 10^{-7}$ \\
\hline
\end{tabular}
\end{table}

\begin{figure}[!h]
  \centering
  \begin{subfigure}[t]{0.49\linewidth}
    \centering
    \includegraphics[width=\linewidth]{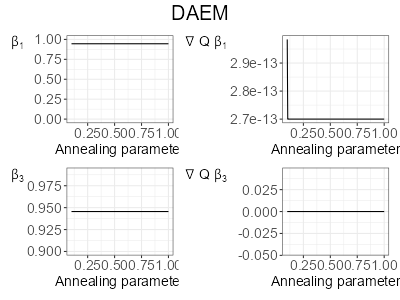}
    \label{fig:daem-trace}
  \end{subfigure}
  \hfill
  \begin{subfigure}[t]{0.49\linewidth}
    \centering
    \includegraphics[width=\linewidth]{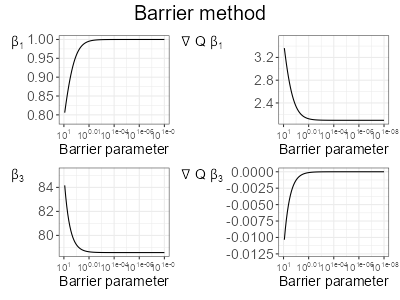}
    \label{fig:bm-trace}
  \end{subfigure}

  \vspace{0.6em}

  \begin{subfigure}[t]{0.49\linewidth}
    \centering
    \includegraphics[width=\linewidth]{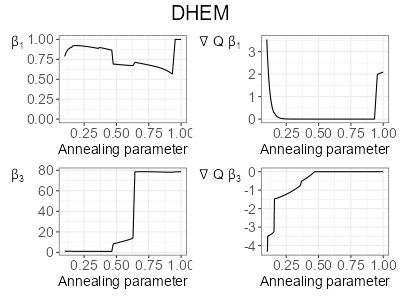}
    \label{fig:dhem-trace}
  \end{subfigure}
  \hfill
  \begin{subfigure}[t]{0.49\linewidth}
    \centering
    \includegraphics[width=\linewidth]{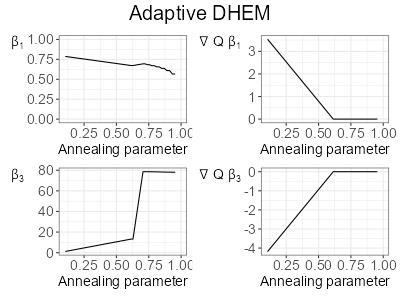}
    \label{fig:adapdhem-trace}
  \end{subfigure}
        \vspace{-12mm}
  \caption{Parameter traces ($\beta$) and gradient traces ($\nabla_\beta$) by four methods.}
  \label{fig:trace-all}
\end{figure}

\Cref{tab:converged_params_math_header} presents the resulting parameter estimates from five methods, and \Cref{fig:trace-all} shows the convergence traces of parameters for the four methods that use hyperparameters. Stationarity can be assessed by verifying that $\nabla Q(\beta_1)$ and $\nabla Q(\beta_3)$ are close to zero. The main distinction among methods lies in how they handle constraint satisfaction and stopping behavior along the homotopy path. In contrast, DHEM enforces the constraint via the barrier term but continues updating according to a fixed schedule, even when stationarity is effectively achieved at intermediate stages. This behavior highlights the advantage of adaptive DHEM, which stops updating at approximately $r=0.954$ when the monotonicity of the observed-data log-likelihood can no longer be maintained. Notably, the parameter estimates obtained by adaptive DHEM are nearly identical to those of DHEM at the point where the differential is minimized, indicating that adaptive stopping successfully captures the effective stationary solution without unnecessary progression along the homotopy path.


\section{Conclusion}

This paper proposes the dual-homotopy EM (DHEM) algorithm and its adaptive variant as a general approach for stable EM-based inference under parameter constraints. Existing constrained EM algorithms often encounter boundary stagnation and lose the monotonicity of the likelihood, which weakens the classical convergence theory of the EM algorithm. We interpret these issues as arising from the interaction between latent-variable distortion, which we refer to as latent effects, and constrained optimization. DHEM combines two homotopy mechanisms: deterministic annealing in the E-step to control the entropy of the latent posterior, and barrier-based optimization in the M-step to enforce parameter constraints smoothly. This approach allows stable exploration within the feasible region during early stages and a gradual transition toward standard EM behavior. Based on this framework, the adaptive DHEM algorithm eliminates reliance on fixed homotopy schedules by adjusting the annealing and barrier parameters in response to monotone improvement in the observed-data likelihood, thus restoring the monotonicity property of EM under constraints. From a theoretical standpoint, we formulate the adaptive DHEM updates through acceptance-rule–based mappings and establish global convergence using Zangwill’s Global Convergence Theorem. As a result, the adaptive DHEM produces sequences that satisfy the constraints, monotonically increase the observed-data likelihood, and converge to stationary points, without depending on distribution-specific arguments. In summary, the adaptive DHEM provides a broadly applicable, theoretically grounded constrained EM framework that is simple to implement and effective across various models. 
%
%
%
Future work may expand this approach to more general nonlinear constraints, high-dimensional settings, and large-scale data applications.

\section*{Appendix}





\subsection*{A.0 Assumptions }
\begin{enumerate}\label{supplement:assumptions}
    \item[(1)] The parameter space $\theta\in\Theta$ is compact.
    \item[(2)] The measure of $X_m$ is finite, i.e, $\int dX_m=\mu_m$.
    \item[(3)] The log-density is integrable with respect to $X_m$, i.e., $\int |\log P(X_o, X_m | \theta)|  dX_m < \infty.$
    \item[(4)] The log-likelihood ratio is integrable with respect to both the measure and the latent posterior distribution. That is, for all $\theta,\theta_1\in\Theta$ $\int\big|\log\big(\frac{P(X_m|X_o,\theta)}{P(X_m|X_o,\theta_1)}\big)\big|P(X_m|X_o,\theta)dX_m<\infty$.
\end{enumerate}


\subsection*{A.1 Some Lemmas and Proofs of Theorems}


\begin{lemma}\label{lemma:taylorZr}
Under the above assumptions (1)--(3), we have the Taylor expansion of $Z_r(\theta)$,
\[
Z_r(\theta)=\int P(X_c|\theta)^r\,dX_m,\quad X_c=(X_o,X_m),\quad r\in[0,1]
\]
as follows.
\begin{enumerate}
    \item For every \(r_0\in(0,1)\), \(Z_r(\theta)\) is differentiable at \(r_0\), with
    \[
    \frac{\partial}{\partial r} Z_r(\theta)\Big|_{r=r_0}=\int P(X_c|\theta)^{r_0}\log P(X_c|\theta)\,dX_m.
    \]
    Thus, it follows that 
    \[
    Z_r(\theta)=Z_{r_0}(\theta)+(r-r_0)\int P(X_c|\theta)^{r_0}\log P(X_c|\theta)\,dX_m+o(|r-r_0|),\quad r\to r_0.
    \]
    \item In particular, at \(r_0=0\),
    \[
    Z_0(\theta)=\int 1\,dX_m=\mu_m,
    \]
    and \(Z_r(\theta)\) admits the first-order expansion
    \[
    Z_r(\theta)=\mu_m+rA(\theta)+o(r),
    \qquad r\downarrow 0,
    \]
    where
    \[
    A(\theta)=\int \log P(X_c|\theta)\,dX_m.
    \]
\end{enumerate}
\end{lemma}
\begin{proof}
Let
$
p_\theta(X_m):=P(X_c|\theta)=P(X_o,X_m|\theta).
$
We first prove differentiability at an interior point \(r_0\in(0,1)\). Fix \(r_0\in(0,1)\) and choose \(\delta>0\) such that
$
[r_0-\delta,r_0+\delta]\subset(0,1).
$
For each \(X_m\),
\[
\frac{\partial}{\partial r}p_\theta(X_m)^r
=
p_\theta(X_m)^r\log p_\theta(X_m).
\]
Thus, it suffices to justify differentiation under the integral sign. For \(t>0\) and \(r\in[r_0-\delta,r_0+\delta]\subset(0,1)\), we bound \(t^r|\log t|\) separately on \(\{t\le 1\}\) and \(\{t>1\}\).
If \(0<t\le 1\), then \(t^r\le 1\), so
$
t^r|\log t|\le |\log t|.
$
When \(t>1\), let
$
\alpha:=1-(r_0+\delta)>0.
$
Since \(\log t\le \frac{1}{\alpha e}t^\alpha\) for all \(t>1\), we have
\[
t^r\log t
\le t^{r_0+\delta}\log t
\le \frac{1}{\alpha e} t^{r_0+\delta+\alpha}
=
\frac{1}{\alpha e} t.
\]
Hence, for all \(r\in[r_0-\delta,r_0+\delta]\), it follows that 
\[
|p_\theta(X_m)^r\log p_\theta(X_m)|
\le
|\log p_\theta(X_m)|\,\mathbf 1_{\{p_\theta\le 1\}}
+
C\,p_\theta(X_m)\mathbf 1_{\{p_\theta>1\}},
\]
where \(C=\frac{1}{\alpha e}\).

Now, the first term is integrable by Assumption (3), since
$
\int |\log p_\theta(X_m)|\,dX_m<\infty.
$
The second term is integrable because
\[
\int p_\theta(X_m)\,dX_m
=
\int P(X_o,X_m|\theta)\,dX_m
=
P(X_o|\theta)
<\infty.
\]
Therefore, 
$
\{\,p_\theta(X_m)^r\log p_\theta(X_m):r\in[r_0-\delta,r_0+\delta]\,\}
$
is dominated by an \(L^1(dX_m)\)-function. By the dominated convergence theorem,
\[
\frac{\partial}{\partial r}Z_r(\theta)\Big|_{r=r_0}
=
\int p_\theta(X_m)^{r_0}\log p_\theta(X_m)\,dX_m.
\]
This proves differentiability for every \(r_0\in(0,1)\), and the first-order Taylor expansion at \(r_0\in(0,1)\) follows.

Next, consider \(r_0=0\). Since
$
Z_0(\theta)=\int p_\theta(X_m)^0\,dX_m=\int 1\,dX_m=\mu_m
$
by Assumption (2), it remains to identify the right derivative at \(0\).
For \(r>0\), by the mean value theorem, for each \(X_m\), there exists \(\xi_r(X_m)\in(0,r)\) such that
\[
\frac{p_\theta(X_m)^r-1}{r}
=
p_\theta(X_m)^{\xi_r(X_m)}\log p_\theta(X_m).
\]
Choose \(r_*>0\) with \(r_*<1\). For \(0<r\le r_*\), the same domination argument as above yields
\[
\left|\frac{p_\theta(X_m)^r-1}{r}\right|
\le
|\log p_\theta(X_m)|\,\mathbf 1_{\{p_\theta\le 1\}}
+
C_*\,p_\theta(X_m)\mathbf 1_{\{p_\theta>1\}},
\]
for some constant \(C_*>0\), and the right-hand side is integrable.
Moreover, pointwise a.e.,
\[
\frac{p_\theta(X_m)^r-1}{r}\to \log p_\theta(X_m)
\qquad (r\downarrow 0).
\]
Hence, by the dominated convergence theorem,
\[
\lim_{r\downarrow 0}\frac{Z_r(\theta)-Z_0(\theta)}{r}
=
\int \log p_\theta(X_m)\,dX_m.
\]
Therefore, 
\[
Z_r(\theta)
=
\mu_m
+
r\int \log P(X_c|\theta)\,dX_m
+
o(r),
\qquad r\downarrow 0.
\]
Letting
\[
A(\theta):=\int \log P(X_c|\theta)\,dX_m
\]
completes the proof.
\end{proof}

\begin{lemma}\label{lemma:Zr}
By \Cref{lemma:taylorZr}, $Z_r(\theta)$ admits a first-order Taylor expansion around $r=0$. Then, 
$
Z_r(\theta)=\mu_m+rA(\theta)+o(r),
$
where
$
A(\theta)=\int \log P(X_c|\theta)dX_m.
$
Hence, it follows that 
\[
\log\frac{Z_r(\theta_1)}{Z_r(\theta_2)}=\frac{r}{\mu_m}\big(A(\theta_1)-A(\theta_2)\big)+o(r).
\]
\end{lemma}

\begin{proof}
Using $a^r=\exp(r\log a)$, we have 
$
P(X_c|\theta)^r=\exp(r\log P(X_c|\theta))=1+r\log P(X_c|\theta)+o(r).
$
Integrating both sides yields
\[
Z_r(\theta)=\int (1+r\log P(X_c|\theta)+o(r))\,dX_m=\mu_m+rA(\theta)+o(r).
\]
Then, applying $\log(1+x)=x+o(x)$ gives the result.
\end{proof}

\begin{lemma}\label{lemma:DKL}
\[
\frac{1}{r}D_{KL}(P_{r},\theta_0\|P_{r},\theta)=O(r).
\]
\end{lemma}

\begin{proof}
By definition, we have 
\begin{align*}
\frac{1}{r}D_{KL}(P_r,\theta_0||P_r,\theta)&=\frac{1}{r}\int \log\bigg(\frac{P_r(X_m|X_o,\theta_0)}{P_r(X_m|X_o,\theta)}\bigg)P_r(X_m|X_o,\theta_0)dX_m\\
    &=\frac{1}{r}\int \log\bigg(\frac{P(X_c|\theta_0)^rZ_r(\theta)}{P(X_c|\theta)^rZ_r(\theta_0)}\bigg)P_r(X_m|X_o,\theta_0)dX_m\\
    &=\int \log\bigg(\frac{P(X_c|\theta_0)}{P(X_c|\theta)}\bigg)P_r(X_m|X_o,\theta_0)dX_m +\frac{1}{r}\log\bigg(\frac{Z_r(\theta)}{Z_r(\theta_0)}\bigg).
\end{align*}
From the following results, 
\[
P_{r}(X_m| X_o,\theta_0)=\frac{\exp\bigl(r\log P(X_c|\theta_0)\bigr)}{Z_r(\theta_0)}, 
\]
$\exp\bigl(r\log P(X_c|\theta_0)\bigr)=1+r\log P(X_c|\theta_0)+o(r)$, and 
$
Z_r(\theta_0)=\mu_m+O(r),
$
we obtain
\[
P_{r}(X_m|X_o,\theta_0)=\frac{1+r\log(P(X_c|\theta_0))+o(r)}{\mu_m+O(r)}. 
\]
Since
\[
\frac{1}{\mu_m+O(r)}=\frac{1}{\mu_m}+O(r),
\]
it follows that
\[
P_{r}(X_m| X_o,\theta_0)=\left(1+r\log(P(X_c|\theta_0))+o(r)\right)\left(\frac{1}{\mu_m}+O(r)\right)=\frac{1}{\mu_m}+O(r). 
\]
Therefore, we have 
\begin{align*}
\int \log\bigg(\frac{P(X_c|\theta_0)}{P(X_c|\theta)}\bigg)P_{r,\theta_0}(X_m| X_o)\,dX_m
&=\int \log\bigg(\frac{P(X_c|\theta_0)}{P(X_c|\theta)}\bigg)\bigg(\frac{1}{\mu_m}+O(r)\bigg)\,dX_m\\
&=\frac{1}{\mu_m}\int \log\bigg(\frac{P(X_c|\theta_0)}{P(X_c|\theta)}\bigg)\,dX_m + O(r)\\
&=\frac{1}{\mu_m}\bigl(A(\theta_0)-A(\theta)\bigr)+O(r), 
\end{align*}
where
$
A(\theta)=\int \log P(X_c|\theta)\,dX_m.
$
On the other hand, by Lemma~\ref{lemma:Zr},
\[
\log\bigg(\frac{Z_r(\theta)}{Z_r(\theta_0)}\bigg)=\frac{r}{\mu_m}\bigl(A(\theta)-A(\theta_0)\bigr)+o(r), 
\]
thus, 
\[
\frac{1}{r}\log\bigg(\frac{Z_r(\theta)}{Z_r(\theta_0)}\bigg)=\frac{1}{\mu_m}\bigl(A(\theta)-A(\theta_0)\bigr)+o(1). 
\]
Then, combining the two expansions provides
\begin{align*}
\frac{1}{r}D_{KL}(P_{r,\theta_0}\|P_{r,\theta})=\frac{1}{\mu_m}\bigl(A(\theta_0)-A(\theta)\bigr)+\frac{1}{\mu_m}\bigl(A(\theta)-A(\theta_0)\bigr)+O(r)
=O(r). 
\end{align*}
\end{proof}

\paragraph{Proof of \Cref{thm:reducingLatentEffect}}\label{pf:reducingLatentEffect}
\begin{proof}
$Q_r(\theta|\theta_0)$ is expressed as 
    \begin{align*}
    &Q_r(\theta|\theta_0) \triangleq l_o(\theta)+G(\theta,\theta_0| r)\\
    &\qquad=\int \log \bigg(P(X_m|X_o,\theta)P(X_o|\theta)\bigg)P_r(X_m|X_o,\theta_0)dX_m\\
    &\qquad=\log P(X_o|\theta)+\int \bigg(\log P(X_m|X_o,\theta)\bigg)P_r(X_m|X_o,\theta_0)dX_m\\
    &\qquad=\log P(X_o|\theta) + \frac{Z(\theta_0)^r}{Z_r(\theta_0)}\int \log P(X_m|X_o,\theta) \frac{P(X_c|\theta_0)^r}{Z(\theta_0)^r} dX_m\\
    &\qquad=\log P(X_o|\theta) + \frac{Z(\theta_0)^r}{Z_r(\theta_0)}\int \log P(X_m|X_o,\theta) P(X_m|X_o,\theta_0)^r dX_m\\
    &\qquad=\log P(X_o|\theta) + \frac{Z(\theta_0)^r}{Z_r(\theta_0)}\int \log P(X_m|X_o,\theta) P(X_m|X_o,\theta_0)^{r-1}P(X_m|X_o,\theta_0) dX_m.
    \end{align*}
    To simplify the expression, we write the terms separately as
    \begin{align*}
    &\int \log P(X_m|X_o,\theta) P(X_m|X_o,\theta_0)^{r-1}P(X_m|X_o,\theta_0) dX_m\\
    &\qquad=E_{P(\cdot|X_o,\theta_0)}\bigg(\log P(X_m|X_o,\theta) P(X_m|X_o,\theta_0)^{r-1}\bigg)\\
    &\qquad=E_{P(\cdot|X_o,\theta_0)}\bigg(\frac{\log P(X_m|X_o,\theta)}{P(X_m|X_o,\theta_0)} P(X_m|X_o,\theta_0)^{r}\bigg).
    \end{align*}
    By the Cauchy-Schwarz inequality, we have 
    \begin{align*}
    &\bigg|E_{P(\cdot|X_o,\theta_0)} \bigg(\frac{\log P(X_m|X_o,\theta)}{P(X_m|X_o,\theta_0)}P(X_m|X_o,\theta_0)^{r}\bigg)\bigg|\\
    &\qquad\le \sqrt{E_{P(\cdot|X_o,\theta_0)}\bigg(\frac{\log P(X_m|X_o,\theta)}{P(X_m|X_o,\theta_0)}\bigg)^2}\sqrt{E_{P(\cdot|X_o,\theta_0)}\bigg(P(X_m|X_o,\theta_0)^{2r}\bigg)}\\
    &\qquad=\sqrt{E_{P(\cdot|X_o,\theta_0)}\bigg(\frac{\log P(X_m|X_o,\theta)}{P(X_m|X_o,\theta_0)}\bigg)^2}\sqrt{\int P(X_m|X_o,\theta_0)^{2r+1}dX_m}\\
    &\qquad=\sqrt{E_{P(\cdot|X_o,\theta_0)}\bigg(\frac{\log P(X_m|X_o,\theta)}{P(X_m|X_o,\theta_0)}\bigg)^2}\sqrt{\frac{Z_{2r+1}(\theta_0)}{Z(\theta_0)^{2r+1}}}. 
    \end{align*}
Combining the above results, we obtain the upper bound of the latent effect $G(\theta,\theta_0| r)$ as 
\begin{align*}
    G(\theta,\theta_0| r)&\le \frac{Z(\theta_0)^r}{Z_r(\theta_0)}\sqrt{E_{P(\cdot|X_o,\theta_0)}\bigg(\frac{\log P(X_m|X_o,\theta)}{P(X_m|X_o,\theta_0)}\bigg)^2}\sqrt{\frac{Z_{2r+1}(\theta_0)}{Z(\theta_0)^{2r+1}}}.
\end{align*}
Note that the following term
$$\sqrt{E_{P(\cdot|X_o,\theta_0)}\bigg(\frac{\log P(X_m|X_o,\theta)}{P(X_m|X_o,\theta_0)}\bigg)^2}$$
does not depend on $r$. By denoting the part depending on $r$ and $\theta_0$ as $h(r|\theta_0)$, we have 
\[
h(r|\theta_0)= \frac{Z(\theta_0)^{2r}}{Z_r(\theta_0)^2}\frac{Z_{2r+1}(\theta_0)}{Z(\theta_0)^{2r+1}}. 
\]
Then, it follows that  
\[
\log h(r|\theta_0)=\log\bigg(\frac{Z_{2r+1}(\theta_0)}{Z(\theta_0)^{2r+1}}\bigg)-2\log \bigg(\frac{Z_r(\theta_0)}{Z(\theta_0)^r}\bigg). 
\]
To obtain an upper bound for $\log h(r|\theta_0)$, we treat the first and second terms of the above expression separately. For the first term, we use the fact that the logarithm function is concave,
\[    \log \bigg(\frac{Z_r(\theta_0)}{Z(\theta_0)^r}\bigg)=\log \bigg( E_{P(X_m|X_o,\theta_0)} P(X_m|X_o,\theta_0)^{r-1}\bigg)\ge (r-1)E_{P(X_m|X_o,\theta_0)}\log P(X_m|X_o,\theta_0). \]
For the second term, we use the concavity of the function for $r<\frac{1}{2}$, 
\[\log \bigg(\frac{Z_{2r+1}(\theta_0)}{Z(\theta_0)^{2r+1}}\bigg)=\log \bigg(E_{P(X_m|X_o,\theta_0)}P(X_m|X_o,\theta_0)^{2r}\bigg)\le 2r \log \bigg(E_{P(X_m|X_o,\theta_0)}P(X_m|X_o,\theta_0)\bigg).\\\]
Thus, the upper bound is given by
\begin{align*}
    \log h(r|\theta_0)\le &~ 2r\bigg( \log\bigg(E_{P(X_m|X_o,\theta_0)}P(X_m|X_o,\theta_0)\bigg)-E_{P(X_m|X_o,\theta_0)}\log P(X_m|X_o,\theta_0)\bigg)\\
    &+2E_{P(X_m|X_o,\theta_0)}\log P(X_m|X_o,\theta_0).
\end{align*}   
By Jensen's inequality, the coefficient of $r$ is nonnegative, that is, 
 \[\log\bigg(E_{P(X_m|X_o,\theta_0)}P(X_m|X_o,\theta_0)\bigg)\ge E_{P(X_m|X_o,\theta_0)}\log P(X_m|X_o,\theta_0).\]
Hence, the upper bound of $\log h(r|\theta_0)$ is monotonically increasing in $r$; therefore, the upper bound of $h(r|\theta_0)$ decreases as $r$ decreases.

\end{proof}

\begin{lemma}[Domination of $P_r$]\label{lemma:dominationPr}
Suppose $\mu_m := \int dX_m < \infty$. Let, for any $r \in (0,1]$,
\[
P_r(X_m|X_o,\theta)
=
\frac{P(X_o,X_m|\theta)^r}{Z_r(\theta)}.
\]
Then, we have 
\[
P_r(X_m|X_o,\theta)\le\frac{1}{c_0(\theta)}\big(1 + P(X_o,X_m|\theta)\big),
\]
where $c_*(\theta) := \inf_{r\in[r_0,1]} Z_r(\theta) > 0$ for any fixed $r_0 \in (0,1]$.
\end{lemma}

\begin{proof}
For $p := P(X_o,X_m|\theta) \ge 0$ and $r \in (0,1]$, we have
$p^r \le 1 + p$. 
Hence, it follows that 
\[
P_r(X_m|X_o,\theta)=\frac{p^r}{Z_r(\theta)}\le\frac{1}{Z_r(\theta)}(1 + p). 
\]
From continuity of $Z_r(\theta)$ and the fact that $Z_1(\theta)=P(X_o|\theta)>0$, 
there exists $c_*(\theta)>0$ such that
$
Z_r(\theta)\ge c_*(\theta)>0~ \text{for all } r \in [r_0,1]
$
Therefore, we obtain 
\[
P_r(X_m|X_o,\theta)\le\frac{1}{c_*(\theta)}\big(1 + P(X_o,X_m|\theta)\big).
\]
\end{proof}

\paragraph{Proof of \Cref{thm:diffKLlowerbound}}\label{pf:diffKLlowerbound}
\begin{proof}
The quantity $\Delta D_{KL}(\theta_1||\theta_0,r)$ can be expressed as  
\begin{align*}
    &\Delta D_{KL}(\theta_1||\theta_0,r) = D_{KL}(P_r,\theta_0||P,\theta_1)-D_{KL}(P_r,\theta_0||P,\theta_0)\\
    &=\int \log \bigg(\frac{P(X_m|X_o,\theta_0)}{P(X_m|X_o,\theta_1)}\bigg)P_r(X_m|X_o,\theta_0)dX_m\\
    &=\int \log \bigg(\frac{P(X_m|X_o,\theta_0)}{P(X_m|X_o,\theta_1)}\bigg)\bigg(P_r(X_m|X_o,\theta_0)-P(X_m|X_o,\theta_0)\bigg)dX_m+D_{KL}(P,\theta_0||P,\theta_1).
\end{align*}
By Assumption (4) in A.0 and \Cref{lemma:dominationPr}, for $r \in [r_0,1]$, we have 
\begin{align*}
\int \bigg|\log \bigg(\frac{P(X_m|X_o,\theta_0)}{P(X_m|X_o,\theta_1)}\bigg)\bigg|
P_r(X_m|X_o,\theta_0)dX_m
&\le \frac{1}{c_0(\theta_0)} 
\int \bigg|\log \bigg(\frac{P(X_m|X_o,\theta_0)}{P(X_m|X_o,\theta_1)}\bigg)\bigg|
\bigg(1+P(X_m|X_o,\theta_0)\bigg)dX_m \\
&\in \mathcal{L}^1.
\end{align*}
Then, applying the dominated convergence theorem (DCT), we obtain
\[
\lim_{r\to1}
\int \log \bigg(\frac{P(X_m|X_o,\theta_0)}{P(X_m|X_o,\theta_1)}\bigg)
\bigg(P_r(X_m|X_o,\theta_0)-P(X_m|X_o,\theta_0)\bigg)dX_m=0.
\]
Therefore, it follows that 
\begin{align*}
\lim_{r\to 1}\Delta D_{KL}(\theta_1,\theta_0|r)= D_{KL}(P,\theta_0|P,\theta_1)= \delta > 0.
\end{align*}
 \end{proof}

\bibliography{reference_short}  
\bibliographystyle{apalike}

\end{document}